\documentclass[aps,prl,twocolumn,groupedaddress,showpacs]{revtex4-1}
\usepackage{amssymb,amsmath,graphicx,epstopdf,hyperref,amsthm}

\begin{document}
\interfootnotelinepenalty=10000 

\title{Deduction of an upper bound on the success probability of port-based teleportation from the no-cloning theorem and the no-signaling principle}

\author{Dami\'an Pital\'ua-Garc\'ia}
\affiliation{Centre for Quantum Information and Foundations, DAMTP, Centre for Mathematical Sciences,
University of Cambridge, Wilberforce Road, Cambridge, CB3 0WA, United Kingdom}

\begin{abstract}
In port-based teleportation, Alice teleports an unknown quantum state $\lvert\psi\rangle$ to one of $N$ ports at Bob's site. Alice applies a measurement and sends Bob the outcome $k$.  Bob only needs to select the $k$th port in order to obtain $\lvert\psi\rangle$. We present a theorem in the spirit of the no-cloning theorem, which says that it is impossible to extract any information from an unknown quantum state if only a single copy of it is provided and if the state remains unchanged. We use this theorem and the no-signaling principle to prove an upper bound on the success probability of port-based teleportation.
\end{abstract}

\maketitle

Quantum teleportation \cite{teleportation} is a fundamental protocol of quantum information theory in which an unknown quantum state $\lvert\psi\rangle$ is destroyed at its original location by Alice and reconstructed at another location by Bob. The protocol works as follows: Alice and Bob must initially share a maximally entangled state, Alice applies a Bell measurement to her systems, she communicates her measurement outcome to Bob, who then applies a unitary correction operation according to Alice's message.

Here, we consider a different type of teleportation protocol, which we refer to as \emph{port-based teleportation} (PBT). PBT was devised by Ishizaka and Hiroshima \cite{IH08,IH09} with the purpose of implementing a universal programmable quantum processor that succeeds with probability arbitrarily close to unity; this task can be achieved using standard teleportation, too, but with a very small success probability for input states of big dimension \cite{NC97}.

In this paper we consider general PBT protocols, which allow Alice to teleport an unknown quantum state $\lvert\psi\rangle$ to one of $N$ ports at Bob's site. PBT requires that Alice and Bob share quantum entanglement and consists of the following steps. Alice applies a measurement with outcome $k\in\lbrace0,1,\dotsc,N\rbrace$; if $k=0$, teleportation fails, otherwise $\lvert\psi\rangle$ is teleported to the $k$th port. Alice communicates $k$ to Bob, who then discards the states at ports with indexes distinct from $k$. No further correction operations are required; this advantage over standard teleportation makes PBT useful in various quantum information tasks. In the probabilistic version of PBT, $\lvert\psi\rangle$ is teleported perfectly but with a success probability $p<1$; in the deterministic version, outcome $k=0$ never occurs but the fidelity of the teleported state is smaller than unity \cite{IH08,IH09}. 

Besides its use as a universal programmable quantum processor, PBT can be used to implement instantaneous nonlocal quantum computation (INLC), reducing exponentially the amount of needed entanglement compared to schemes based on standard teleportation \cite{BK11}. INLC is the application of a nonlocal unitary operation $U$ on a state $\lvert\psi\rangle$ shared by two or more distant parties with a single round of classical communication (CC). If two rounds of CC are allowed, $U$ can be implemented trivially as follows: Alice teleports her part of $\lvert\psi\rangle$ to Bob, who then applies $U$ to $\lvert\psi\rangle$, now in his location, and then teleports Alice's share of the state back to her. However, it is not trivial to complete this task with only one round of CC. It was Vaidman \cite{V03} who first showed how to implement this task using a recursive scheme based on standard teleportation, which consumes an amount of entanglement growing double exponentially with the number of qubits $n$ of the input state $\lvert\psi\rangle$. However, a scheme by Beigi and K{\"{o}}nig \cite{BK11} based on PBT allows the implementation of INLC with an amount of entanglement growing only exponentially with $n$. 

INLC has application to other distributed quantum tasks: it allows the implementation of instantaneous nonlocal measurements (INLM) and also breaks the security of position-based quantum cryptography (PBQC) and some quantum tagging schemes \cite{KBMS06,M10.1,*M10.2,KMS11,LL11,pbqc,BFSS11}. INLM is the measurement of a nonlocal observable in a distributed quantum state with a single round of CC \cite{GR02,V03,GRV03,CCJP10}. Quantum tagging  \cite{KBMS06,M10.1,*M10.2,KMS11,K11.1} and PBQC \cite{pbqc,BFSS11} are cryptographic tasks that rely on quantum information processing and relativistic constraints with the goals of verifying the location of an object and providing secure communication with a party at a given location, respectively.

It is interesting to investigate the limitations and possibilities on quantum information processing tasks that can be derived directly from no-signaling (NS) and other fundamental principles of quantum theory. Some important results obtained with this approach are the following: the maximum fidelity achieved by quantum cloning machines can be deduced from NS \cite{G98}, security of quantum key distribution can be guaranteed as long as NS is satisfied \cite{BHK05}, and the recently discovered information causality principle \cite{ic} implies Tsirelson's bound on the amount of nonlocality for quantum correlations.

In this paper, we show an upper bound on the success probability of probabilistic PBT of a $n$-qubit state from the no-cloning theorem and the no-signaling principle:
\begin{equation}
\label{eq:2}
p\leq\frac{N}{4^n+N-1}.
\end{equation}
Our bound agrees with the maximum success probability obtained in \cite{IH09} for the particular case $n=1$:
\begin{equation}
\label{eq:1}
p_{\max}=\frac{N}{3+N}.
\end{equation}
Thus, we confirm the hypothesis presented in \cite{IH09} that Eq.~(\ref{eq:1}) can be derived from fundamental laws of physics. It is an interesting open problem to find a probabilistic PBT protocol for the case $n>1$ and to see whether our bound is achievable.

Comparing Eqs.~(\ref{eq:2}) and (\ref{eq:1}), we see that $(p_{\max})^n$ can be bigger than the upper bound on $p$, which means that applying PBT individually to each qubit of $\lvert\psi\rangle$ can give a higher success probability than applying PBT globally to $\lvert\psi\rangle$. However, we justify the restriction that $\lvert\psi\rangle$ must be localized to a single port by noting that the advantage of PBT as described here, at least for implementing a universal programmable quantum processor and INLC, is that, before receiving Alice's message, Bob can apply the desired quantum operation to the state at every port, after which, $\lvert\psi\rangle$ is transformed as desired. Clearly, this advantage is lost if the qubits of $\lvert\psi\rangle$ spread among different ports.

Before summarizing our proof of Eq.~(\ref{eq:2}), it will be useful to give a general description of the PBT protocol. For simplicity of the exposition we consider a pure input state $\lvert\psi\rangle_a\in\mathcal{H}_a$. Due to the linearity of quantum theory, the protocol works for mixed states too. Alice and Bob share a fixed entangled state $\lvert\xi\rangle_{AB}\in\mathcal{H}_A\otimes\mathcal{H}_B$, which is independent of $\lvert\psi\rangle$ because this is arbitrary and unknown. Bob has $N$ ports $\lbrace B_j\rbrace_{j=1}^N$, hence $\mathcal{H}_B=\bigotimes\limits_{j=1}^N\mathcal{H}_{B_j}$, where $\text{dim}\mathcal{H}_{B_j}=\text{dim}\mathcal{H}_a=2^n~\forall j\in\lbrace1,\dotsc,N\rbrace$. System $A$ includes any ancilla held by Alice and so has an arbitrarily big dimension. However, in \cite{IH08,IH09}, $\mathcal{H}_A=\bigotimes\limits_{j=1}^N\mathcal{H}_{A_j}$ and $\text{ dim}\mathcal{H}_{A_j}=\text{dim}\mathcal{H}_a~\forall j\in\lbrace1,\dotsc,N\rbrace$. We follow a notation in which subindex $a$ is written in $\lvert\psi\rangle_a$ only when we wish to
emphasize that the system $a$ is in state $\lvert\psi\rangle$, similarly for other states and systems. The initial global state is
\begin{equation}
\label{eq:7}
\lvert G\rangle_{aAB}=\lvert\psi\rangle_{a}\lvert\xi\rangle_{AB}.
\end{equation}
Alice applies a generalized measurement, which in general can be decomposed into a unitary operation $U$ acting jointly on $a$ and $A$, followed by a projective measurement. Alice obtains outcome $k\in\lbrace1,\dotsc,N\rbrace$ with probability $q_k>0$ and $k=0$ with probability $1-\sum_{k=1}^N q_k$. If $k\neq0$, the global state is transformed into
\begin{equation}
\label{eq:8}
\lvert G_k\rangle_{aAB}=\lvert\psi\rangle_{B_k}\lvert R_k\rangle_{aA\tilde{B}_k},
\end{equation}
where $\tilde{B}_k\equiv B_1B_2\dotsm B_{k-1}B_{k+1}B_{k+2}\dotsm B_N$, so state $\lvert\psi\rangle$ is teleported to port $B_k$. However, if $k=0$, PBT fails; in this case we denote the final state as
\begin{equation}
\label{eq:8.5}
\lvert G_0\rangle_{aAB}=\lvert F^{(\psi)}\rangle_{aAB}.
\end{equation}
The total success probability is
\begin{equation}
\label{eq:4}
p\equiv\sum_{j=1}^N q_j.
\end{equation}

Now we are able to summarize the proof. First, we present a version of the no-cloning theorem, which allows us to show that the probabilities $q_k$ and states $\lvert R_k\rangle$ cannot depend on $\lvert\psi\rangle$, while the state $\lvert F^{(\psi)}\rangle$ must do, as the notation suggests. Second, we use the no-signaling principle to show that the state $\hat{\eta}_j$ of port $B_j$ before implementing PBT must be of the form
\begin{equation}
\label{eq:3}
\hat{\eta}_j=q_j\lvert\psi\rangle\langle\psi\lvert+\sum_{\substack{i=1\\i\neq j}}^N q_i\hat{\gamma}_{j,i}+(1-p)\hat{\omega}_j^{(\psi)},
\end{equation}
where $\hat{\gamma}_{j,i}$ and $\hat{\omega}_j^{(\psi)}$ are the states to which $B_j$ transforms into after outcomes $k=i\notin\lbrace 0,j\rbrace$ and $k=0$ are obtained, respectively. Since $\hat{\eta}_j$, $\hat{\gamma}_{j,i}$ and $\hat{\omega}_j^{(\psi)}$ are reduced states of $\lvert\xi\rangle$, $\lvert R_i\rangle$ and $\lvert F^{(\psi)}\rangle$, respectively, $\hat{\eta}_j$ and $\hat{\gamma}_{j,i}$ do not depend on $\lvert\psi\rangle$, while $\hat{\omega}_j^{(\psi)}$ does. Third, we use the independence of these states from $\lvert\psi\rangle$ and Eq.~(\ref{eq:3}) to show that  if there exists a protocol that achieves success probability $q_j$  for some states $\hat{\eta}_j$ and $\hat{\gamma}_{j,i}$ then there exists a protocol satisfying
\begin{equation}
\label{eq:5}
\hat{\eta}_j=\hat{\gamma}_{j,i}=\frac{I}{2^n},
\end{equation}
which achieves the same success probability. Fourth, we assume Eq.~(\ref{eq:5}) and present a protocol in which Alice tries to send Bob a random message of $2n$ bits. This protocol combines the superdense coding protocol \cite{sdc} and a modified PBT protocol in which Alice holds every port except for $B_j$, which is held by Bob, but does not allow communication. We show that this protocol succeeds with probability
\begin{equation}
\label{eq:6}
p_j'=q_j+\frac{1}{4^n}(p-q_j)+(1-p)r_j,
\end{equation}
for some probability $r_j$. Since there is not communication in such a protocol, the no-signaling principle implies that Bob cannot obtain any information about Alice's message. This means that Bob can only obtain the correct message with the probability of making a random guess: $p_j'=4^{-n}$ \footnote{It is not possible that $p_j'<4^{-n}$, because this would imply that a modified protocol in which Bob applies a permutation to the obtained message succeeds with a probability higher than $4^{-n}$, violating the no-signaling principle.}. Thus, we have
\begin{equation}
\label{eq:6.5}
q_j+\frac{1}{4^n}(p-q_j)+(1-p)r_j=\frac{1}{4^n}.
\end{equation}
Summing over $j\in\lbrace1,2,\dotsc,N\rbrace$ and using Eq.~(\ref{eq:4}), we obtain that $p=f_{n,N}(R)$, where $R\equiv \sum_{j=1}^{N}r_j$ and
$f_{n,N}(R)\equiv \Bigl(1+\frac{4^n-1}{N-4^nR}\Bigr)^{-1}$. It is straightforward to obtain that the condition $0\leq p\leq 1$ is satisfied only if $R\leq 4^{-n}N$ . Since the function $f_{n,N}(R)$ decreases monotonically with $R$ in the range $[0,4^{-n}N]$, we have that $f_{n,N}(R)\leq f_{n,N}(0)=N/(N+4^n-1)$. Thus, we obtain the desired bound,
\begin{equation}
p\leq\frac{N}{4^n+N-1}.\nonumber
\end{equation}
We continue with the arguments of the proof.

The following theorem is in the spirit of the no-cloning theorem \cite{WZ82,D82}, in a probabilistic \cite{B07} and a stronger \cite{J02,*J03} version, and tells us that it is impossible to extract any information from a single copy of an unknown quantum state without modifying it.
\newtheorem*{Theorem}{Theorem}
\begin{Theorem}
\label{Theorem}
Consider a single copy of an unknown pure quantum state $\lvert\psi\rangle_a\in\mathcal{H}_a$ and a fixed initial state $\lvert\xi\rangle_b\in\mathcal{H}_b$ of an auxiliary system of arbitrary dimension. A physical operation $O$ that induces a transformation $T_k$:
\begin{equation}
\lvert\psi\rangle_a\lvert\xi\rangle_b\longrightarrow\lvert\psi\rangle_a\lvert R_k^{(\psi)}\rangle_b,\nonumber
\end{equation}
with probability $q_k^{(\psi)}>0$, for $k\in\lbrace1,2,\dotsc,N\rbrace$, and a transformation $T_0$:
\begin{equation}
\lvert\psi\rangle_a\lvert\xi\rangle_b\longrightarrow\lvert F^{(\psi)}\rangle_{ab},\nonumber
\end{equation}
with probability $1-\sum_{k=1}^N q_k^{(\psi)}$, for all $\lvert\psi\rangle_a\in\mathcal{H}_a$, in which the index $j\in\lbrace0,1,\dotsc,N\rbrace$ of the induced transformation $T_j$ is known after $O$ is completed, is possible only if
\begin{equation}
q_k^{(\psi)}=q_k^{(\phi)},\quad\lvert R_k^{(\psi)}\rangle_b=\lvert R_k^{(\phi)}\rangle_b,\quad\langle F^{(\phi)}\vert F^{(\psi)}\rangle=\langle\phi\vert\psi\rangle,\nonumber
\end{equation}
for all $\lvert\psi\rangle_a,\lvert\phi\rangle_a\in\mathcal{H}_a$ and $k\in\lbrace1,2,\dotsc,N\rbrace$.
\end{Theorem}

The proof of the theorem is presented in the Supplemental Material. The theorem implies the following lemma.

\newtheorem*{Lemma}{Lemma}
\begin{Lemma}
\label{Lemma}
In a PBT protocol, as described by Eqs.~(\ref{eq:7})--(\ref{eq:8.5}), for every input state $\lvert\psi\rangle_a$ the following is true. The probability $q_k$ of successful teleportation to port $B_k$ and the residual state $\lvert R_k\rangle_{aA\tilde{B}_k}$ when $\lvert\psi\rangle$ is teleported to port $B_k$ do not depend on $\lvert\psi\rangle$. However, the global state $\lvert F^{(\psi)}\rangle_{aAB}$ obtained after a failed PBT depends on $\lvert\psi\rangle$.
\end{Lemma}
\begin{proof}
We can identify the physical operation $O$ in the theorem with the PBT protocol described by Eqs.~(\ref{eq:7})-- (\ref{eq:8.5}) followed by a swap operation of systems $a$ and $B_k$ when outcome $k\neq 0$ is obtained. Therefore, according to the theorem, the probability $q_k$ and the state $\lvert R_k\rangle$ in Eq.~(\ref{eq:8}) do not depend on $\lvert\psi\rangle$, while the state $\lvert F^{(\psi)}\rangle$ in Eq.~(\ref{eq:8.5}) does.
\end{proof}

The fact that the states $\lvert R_k\rangle_{aA\tilde{B}_k}$ do not depend on $\lvert\psi\rangle$, together with the fact that Alice knows the resource state $\lvert\xi\rangle_{AB}$, her unitary $U$, her projective measurement, and its result $k$, implies that Alice knows the states $\lvert R_k\rangle_{aA\tilde{B}_k}$. This will be useful later.

Now we prove Eq.~(\ref{eq:3}). We consider the state $\hat{\eta}_j$ of system $B_j$ held by Bob before the PBT protocol begins. From Eq.~(\ref{eq:7}) we have that
\begin{eqnarray}
\label{eq:8.6}
\hat{\eta}_j&&\equiv\text{Tr}_{aA\tilde{B}_j} \left(\lvert G\rangle\langle G\lvert\right)_{aAB}\nonumber\\&&=\text{Tr}_{A\tilde{B}_j} \left(\lvert\xi\rangle\langle\xi\lvert\right)_{AB}.
\end{eqnarray}
The no-signaling principle implies that from Bob's point of view, his state does not change with Alice's LO if she does not send him any information. However, from Alice's point of view, after she applies her LO, Bob's state changes according to her measurement result $k$. 

1) With probability $q_j, k=j$ and $\hat{\eta}_j$ changes to $\lvert\psi\rangle$, as can be seen from Eq.~(\ref{eq:8}):
\begin{equation}
\label{eq:8.7}
\text{Tr}_{aA\tilde{B}_j} \left(\lvert G_j\rangle\langle G_j\lvert\right)_{aAB}=\left(\lvert\psi\rangle\langle\psi\lvert\right)_{B_j}.\nonumber
\end{equation}

2) With probability $q_i$, $k=i\notin\lbrace0,j\rbrace$ and $\hat{\eta}_j$ changes to some state that we denote as $\hat{\gamma}_{j,i}$. From Eq.~(\ref{eq:8}) we have that
\begin{eqnarray}
\label{eq:8.8}
\hat{\gamma}_{j,i}&&\equiv\text{Tr}_{aA\tilde{B}_j} \left(\lvert G_i\rangle\langle G_i\lvert\right)_{aAB}\nonumber\\&&=\text{Tr}_{aA\tilde{B}_{j,i}} \left(\lvert R_i\rangle\langle R_i\lvert\right)_{aA\tilde{B}_i},
\end{eqnarray}
where $\tilde{B}_{j,i}\equiv B_1\dotsm B_{j-1}B_{j+1}\dotsm B_{i-1}B_{i+1}\dotsm B_N$. In this case $\lvert\psi\rangle$ is successfully teleported to port $B_i\neq B_j$. 

3) With probability $1-p, k=0$ and $\hat{\eta}_j$ changes to some state that we call $\hat{\omega}_j^{(\psi)}$. From Eq.~(\ref{eq:8.5}) we have that
\begin{eqnarray}
\label{eq:8.9}
\hat{\omega}_j^{(\psi)}&&\equiv\text{Tr}_{aA\tilde{B}_{j}} \left(\lvert G_0\rangle\langle G_0\lvert\right)_{aAB}\nonumber\\&&=\text{Tr}_{aA\tilde{B}_{j}} \left(\lvert F^{(\psi)}\rangle\langle F^{(\psi)}\lvert\right)_{aAB}.
\end{eqnarray}
This is the failure result, hence $\hat{\omega}_j^{(\psi)}\neq\lvert\psi\rangle\langle\psi\lvert$ in general. 

Since $\lvert\xi\rangle$ is fixed, we see from  Eq.~(\ref{eq:8.6}) that $\hat{\eta}_j$ does not depend on $\lvert\psi\rangle$. Equations~(\ref{eq:8.8}) and (\ref{eq:8.9}) and the lemma imply that $q_j$ and $\hat{\gamma}_{j,i}$ do not depend on $\lvert\psi\rangle$, while $\hat{\omega}_j^{(\psi)}$ does. Due to the no-signaling principle, Bob cannot learn Alice's outcome before he receives any information from her. Therefore, from Bob's point of view, before receiving any information from Alice, his state is
\begin{equation}
\hat{\eta}_j=q_j\lvert\psi\rangle\langle\psi\lvert+\sum_{\substack{i=1\\i\neq j}}^N q_i\hat{\gamma}_{j,i}+(1-p)\hat{\omega}_j^{(\psi)},\nonumber
\end{equation}
which is Eq.~(\ref{eq:3}).

Now we show that if there exists a protocol that achieves success probability $q_j$  for some states $\hat{\eta}_j$ and $\hat{\gamma}_{j,i}$ then there exists a protocol with corresponding states $\hat{\eta}_j', \hat{\gamma}_{j,i}'$ and $\hat{\omega}_j'^{(\psi)}$ that achieves the same success probability and satisfies $\hat{\eta}_j' = \hat{\gamma}_{j,i}' = \frac{I}{2^n}$, that is, Eq.~(\ref{eq:5}). The claimed protocol, which for convenience we call \emph{primed}, is the following (see details in the Supplemental Material).

We define the set of unitary operations $\lbrace V_l\rbrace_{l=1}^{4^n}\equiv\lbrace\sigma_0,\sigma_1,\sigma_2,\sigma_3\rbrace^{\otimes n}$, where $\sigma_0$ is the identity acting on $\mathbb{C}^2$ and $\lbrace\sigma_i\rbrace_{i=1}^3$ are the Pauli matrices. We will use the identity
\begin{equation}
\label{eq:id}
\frac{I}{2^n}\equiv \frac{1}{4^n}\sum_{l=1}^{4^n} V_l\hat{\rho} V_l^\dagger,
\end{equation}
which is satisfied for any quantum state $\hat{\rho}$ of dimension $2^n$. Consider an ancilla $a'$ with Hilbert space $\mathcal{H}_{a'}$ of dimension $4^n$ at Alice's site, which in general can be included as part of $A$, but which for clarity of the presentation is distinguished as a different system. Let $\lbrace\lvert\mu_l\rangle\rbrace_{l=1}^{4^{n}}$ be an orthonormal basis of $\mathcal{H}_{a'}$. The ancilla $a'$ is prepared in the state $\lvert\phi\rangle\equiv\frac{1}{2^n}\sum_{l=1}^{4^n}\lvert\mu_l\rangle$. Conditioned on $a'$ being in the state $\lvert\mu_l\rangle$, the following operations are performed. Before implementing PBT, Bob's system $B_j$ is prepared in the state $V_l\hat{\eta}_jV_l^{\dagger}$. Then, if Alice applies the PBT protocol described by Eqs.~(\ref{eq:7})--(\ref{eq:8.5}), which satisfies Eq.~(\ref{eq:3}), on her system $aA$ then with probability $q_j$ the state of system $B_j$ transforms into $V_l\lvert\psi\rangle$ and with probability $q_i$ transforms into $V_l\hat{\gamma}_{j,i}V_l^{\dagger}$, where $i\notin\lbrace j,0\rbrace$; this is clear from Eq.~(\ref{eq:3}) and follows from the fact that the operations on $B_j$ commute with those on $aA$ (no-signaling) and from the linearity of quantum theory (see details in the Supplemental Material). Thus, consider that before doing this, Alice applies $V_l^{\dagger}$ on her input state $\lvert\psi\rangle_a$. In this case, the state $\lvert\psi\rangle$ is teleported without error. Since the states of system $B_j$ before implementing PBT and after an outcome $k=i\notin\lbrace j,0\rbrace$ is obtained do not depend on the teleported state, these states remain the same. Hence, in the primed PBT protocol we obtain that after discarding the ancilla $a'$, by taking the partial trace over $\mathcal{H}_{a'}$, the initial state of system $B_j$ is $\hat{\eta}_j'=\frac{1}{4^n}\sum_{l=1}^{4^n}V_l\hat{\eta}_jV_l^{\dagger}$ and its final state after an outcome $k=i\notin\lbrace j,0\rbrace$ is obtained is $\hat{\gamma}_{j,i}'=\frac{1}{4^n}\sum_{l=1}^{4^n}V_l\hat{\gamma}_{j,i}V_l^{\dagger}$, both of which equal $\frac{I}{2^n}$, as follows from Eq.~(\ref{eq:id}). Therefore, we see from Eq.~(\ref{eq:3}) that this protocol satisfies
\begin{equation}
\label{eq:11}
\frac{I}{2^n}=q_j\lvert\psi\rangle\langle\psi\lvert+\sum_{\substack{i=1\\i\neq j}}^N q_i\frac{I}{2^n}+(1-p)\hat{\omega}_j'^{(\psi)},
\end{equation}
where
\begin{equation*}
\hat{\omega}_j'^{(\psi)}\equiv\frac{1}{4^n}\sum_{l=1}^{4^n}V_l\hat{\omega}_j^{(\psi_l)}V_l^{\dagger},
\end{equation*}
and $\psi_l$ refers to dependence on the state $V_l^{\dagger}\lvert\psi\rangle$.

We have shown that the previous PBT protocol succeeds with probability $q_j$ and satisfies Eq.~(\ref{eq:5}). Thus, we assume Eq.~(\ref{eq:5}) in what follows, and without loss of generality we consider that the system $a'$ is included in $A$, hence we do not need to mention it again.

Now we present a protocol in which Alice tries to send Bob a random message of $2n$ bits that succeeds with probability $p_j'$, given by Eq.~(\ref{eq:6}). Note that so far we have considered the input system $a$ to be in a pure state. However, the previous arguments work if $a$ is in a mixed state too. Thus, consider that $a$ is in a bipartite maximally entangled state with system $b$, held by Bob, and that Alice and Bob perform superdense cding (SDC) \cite{sdc} using this state. Alice applies a local unitary on $a$ that encodes a message of $2n$ bits, after which the system $ab$ is in the state $\lvert\psi\rangle_{ab}$. But, instead of sending system $a$ directly to Bob, Alice teleports its state using the modified PBT protocol described below in which no communication is allowed. Bob completes SDC by measuring the system $B_jb$ in an orthonormal basis that includes the state $\lvert\psi\rangle$. Bob obtains Alice's message correctly if his outcome corresponds to the state $\lvert\psi\rangle$.

Bob has the system $B_jb$ and Alice has the system $aA\tilde{B}_j$. Similar to the original PBT protocol, the initial global state is given by $\lvert\psi\rangle_{ab}\lvert\xi\rangle_{AB}$, and Alice applies the same local operations (LO) on systems $a$ and $A$, only. Therefore, if Alice's measurement result is $k\neq 0$, the final state is $\lvert\psi\rangle_{B_kb}\lvert R_k\rangle_{aA\tilde{B}_k}$ and the residual state $\lvert R_k\rangle_{aA\tilde{B}_k}$ is known by her. However, \emph{Alice is not allowed to communicate with Bob}. From Eq.~(\ref{eq:5}), we assume that Bob's register $B_j$ is initially in the completely mixed state, meaning that it is maximally entangled with its purifying system, at Alice's site. Consider the possible situations according to Alice's outcome $k$.

1) Alice obtains $k=j$, so the system $B_jb$ is transformed into the state $\lvert\psi\rangle_{B_jb}$; this occurs with probability $q_j$. Bob applies the SDC measurement and obtains Alice's message correctly.

2) Alice obtains $k=i\notin\lbrace0,j\rbrace$, so the state of system $a$ is teleported to the register $B_i$ at Alice's site; this occurs with probability $q_i$. We denote the composite system $aA\tilde B_{j,i}$ as $A_{j,i}$. Alice has the systems $A_{j,i}$ and $B_i$, while Bob has the system $B_jb$. The global system is in the state $\lvert\psi\rangle_{B_ib}\lvert R_i\rangle_{A_{j,i}B_j}$. Equation~(\ref{eq:5}) tells us that the system $B_j$ is completely mixed, meaning that it is maximally entangled with its purifying system $A_{j,i}$,
\begin{equation}
\lvert R_i\rangle_{A_{j,i}B_j}=\frac{1}{\sqrt{2^n}}\sum_{l=1}^{2^n}\lvert l_i \rangle_{A_{j,i}}\lvert l_i \rangle_{B_j}\nonumber,
\end{equation}
where $\lbrace\lvert l_i\rangle\rbrace_{l=1}^{2^n}$ is the Schmidt basis of $\lvert R_i\rangle_{A_{j,i}B_j}$, which is known by Alice. Therefore, Alice can apply the LO of the standard teleportation protocol \cite{teleportation} to the systems $B_i$ and $A_{j,i}$ in order to teleport the state of system $B_i$ to system $B_j$. Then, Bob applies the SDC measurement on $B_jb$. Since communication is not allowed, the no-signaling principle implies that Bob obtains Alice's message correctly with probability $4^{-n}$. Thus, if $k\notin\lbrace0,j\rbrace$, the success probability is
\begin{equation}
\frac{1}{4^n}\sum_{\substack{i=1\\i\neq j}}^N q_i=\frac{1}{4^n}(p-q_j).\nonumber
\end{equation}

3) Alice obtains $k=0$, hence the protocol fails; this occurs with probability $1-p$. In this case, we should allow for the possibility that the final state of the system $B_jb$ has nonzero overlap with the input state $\lvert\psi\rangle$. Therefore, after applying the SDC measurement, the system $B_jb$ transforms into the state $\lvert\psi\rangle_{B_jb}$ with some probability $r_j$. Thus, if $k=0$, Bob obtains Alice's message with probability $r_j$.

The total success probability $p_j'$ of the previous protocol is given by Eq.~(\ref{eq:6}):
\begin{equation}
p_j'=q_j+\frac{1}{4^n}(p-q_j)+(1-p)r_j.\nonumber
\end{equation}
Thus, Eq.~(\ref{eq:6.5}) and our main result, Eq.~(\ref{eq:2}), follow.

In summary, we have shown an upper bound on the success probability $p$ of probabilistic port-based teleportation (PBT) of an unknown quantum state of $n$ qubits as a function of $n$ and the number of ports $N$, Eq.~(\ref{eq:2}). Our proof is based on the no-signaling principle and a version of the no-cloning theorem, which we have presented in this paper. Our bound on $p$ agrees with the maximum success probability for the case $n=1$ \cite{IH09}. A probabilistic PBT protocol for the case $n>1$ has not been developed explicitly; it would be interesting to know whether our bound can be achieved in this case too.

\begin{acknowledgments}
I would like to thank Adrian Kent for much assistance with this work and Tony Short for helpful discussions. I acknowledge financial support from CONACYT M\'exico and partial support from Gobierno de Veracruz. 
\end{acknowledgments}

\section{Supplemental Material}
\subsection{Proof of the theorem}
\newtheorem*{TheoremApp}{Theorem}
\begin{TheoremApp}
\label{TheoremApp}
Consider a single copy of an unknown pure quantum state $\lvert\psi\rangle_a\in\mathcal{H}_a$ and a fixed initial state $\lvert\xi\rangle_b\in\mathcal{H}_b$ of an auxiliary system of arbitrary dimension. A physical operation $O$ that induces a transformation $T_k$:
\begin{equation}
\lvert\psi\rangle_a\lvert\xi\rangle_b\longrightarrow\lvert\psi\rangle_a\lvert R_k^{(\psi)}\rangle_b,\nonumber
\end{equation}
with probability $q_k^{(\psi)}>0$, for $k\in\lbrace1,2,\dotsc,N\rbrace$, and a transformation $T_0$:
\begin{equation}
\lvert\psi\rangle_a\lvert\xi\rangle_b\longrightarrow\lvert F^{(\psi)}\rangle_{ab},\nonumber
\end{equation}
with probability $1-\sum_{k=1}^N q_k^{(\psi)}$, for all $\lvert\psi\rangle_a\in\mathcal{H}_a$, in which the index $j\in\lbrace0,1,\dotsc,N\rbrace$ of the induced transformation $T_j$ is known after $O$ is completed, is possible only if
\begin{equation}
q_k^{(\psi)}=q_k^{(\phi)},\quad\lvert R_k^{(\psi)}\rangle_b=\lvert R_k^{(\phi)}\rangle_b,\quad\langle F^{(\phi)}\vert F^{(\psi)}\rangle=\langle\phi\vert\psi\rangle,\nonumber
\end{equation}
for all $\lvert\psi\rangle_a,\lvert\phi\rangle_a\in\mathcal{H}_a$ and $k\in\lbrace1,2,\dotsc,N\rbrace$.
\end{TheoremApp}
\begin{proof}
In general, the physical operation $O$ can be decomposed into a unitary operation $U$ acting on the input system and an ancilla of sufficiently big dimension, followed by a projective measurement. Without loss of generality we can consider that the measurement is made on a pointer system $\pi$ of dimension $N+1$, with the outcome indicating the induced transformation. Let $\lvert\chi\rangle_\pi\in\mathcal{H}_{\pi}$ be a fixed initial state of $\pi$ and $\lbrace\lvert k\rangle_{\pi}\rbrace_{k=0}^N$  be an orthonormal basis of $\mathcal{H}_{\pi}$. $U$ must be such that
\begin{eqnarray}
\label{eq:a1}
U\lvert\psi\rangle_a\lvert\xi\rangle_b\lvert\chi\rangle_{\pi}=&&\sqrt{1-p^{(\psi)}}\lvert F^{(\psi)}\rangle_{ab}\lvert 0\rangle_{\pi}\nonumber\\
&&+\sum_{k=1}^N\sqrt{q_k^{(\psi)}}\lvert\psi\rangle_a\lvert R_k^{(\psi)}\rangle_b\lvert k\rangle_{\pi},\nonumber\\
\end{eqnarray}
where $p^{(\psi)}\equiv\sum_{k=1}^N q_k^{(\psi)}$. Let $\mathcal{H}_a$ have dimension $d$, $\lbrace\lvert e_l\rangle_a\rbrace_{l=1}^d$ be an orthonormal basis of $\mathcal{H}_a$ and the expansion of $\lvert\psi\rangle_a$ in this basis be
\begin{equation}
\label{eq:a2}
\lvert\psi\rangle_a=\sum_{l=1}^d \alpha_l\lvert e_l\rangle_a.
\end{equation}
Applying Eq.~(\ref{eq:a1}) to state $\lvert e_l\rangle_a$ we have
\begin{eqnarray}
\label{eq:a3}
U\lvert e_l\rangle_a\lvert\xi\rangle_b\lvert\chi\rangle_{\pi}=&&\sqrt{1-p^{(e_l)}}\lvert F^{(e_l)}\rangle_{ab}\lvert 0\rangle_{\pi}\nonumber\\
&&+\sum_{k=1}^N\sqrt{q_k^{(e_l)}}\lvert e_l\rangle_a\lvert R_k^{(e_l)}\rangle_b\lvert k\rangle_{\pi}.\nonumber\\
\end{eqnarray}
On the one hand, Eqs.~(\ref{eq:a2}), (\ref{eq:a3}) and the linearity of unitary evolution imply that
\begin{eqnarray}
\label{eq:a4}
U\lvert\psi\rangle_a\lvert\xi\rangle_b\lvert\chi\rangle_{\pi}=&&\sum_{l=1}^d\alpha_l\sqrt{1-p^{(e_l)}}\lvert F^{(e_l)}\rangle_{ab}\lvert 0\rangle_{\pi}\nonumber\\
&&+\sum_{l=1}^d\alpha_l\sum_{k=1}^N\sqrt{q_k^{(e_l)}}\lvert e_l\rangle_a\lvert R_k^{(e_l)}\rangle_b\lvert k\rangle_{\pi}.\nonumber\\
\end{eqnarray}
On the other hand, Eqs.~(\ref{eq:a1}) and (\ref{eq:a2}) imply
\begin{eqnarray}
\label{eq:a5}
U\lvert\psi\rangle_a\lvert\xi\rangle_b\lvert\chi\rangle_{\pi}=&&\sqrt{1-p^{(\psi)}}\lvert F^{(\psi)}\rangle_{ab}\lvert 0\rangle_{\pi}\nonumber\\
&&+\sum_{k=1}^N\sum_{l=1}^d\alpha_l\sqrt{q_k^{(\psi)}}\lvert e_l\rangle_a\lvert R_k^{(\psi)}\rangle_b\lvert k\rangle_{\pi}.\nonumber\\
\end{eqnarray}
Since $\lbrace\lvert k\rangle_{\pi}\rbrace_{k=0}^N$ and $\lbrace\lvert e_l\rangle_a\rbrace_{l=1}^d$ are orthonormal bases, it is straightforward to see that both Eqs.~(\ref{eq:a4}) and (\ref{eq:a5}) can be valid only if
\begin{equation}
\label{eq:a6}
q_k^{(e_l)}=q_k^{(\psi)},\qquad\lvert R_k^{(e_l)}\rangle_b=\lvert R_k^{(\psi)}\rangle_b,
\end{equation}
for all $l\in\lbrace1,\dotsc,d\rbrace$ and all $k\in\lbrace1,\dotsc,N\rbrace$.  Since this result is independent of the input state $\lvert\psi\rangle_a$ and of its orthogonal decomposition, it must be that $q_k^{(\psi)}$ and $\lvert R_k^{(\psi)}\rangle_b$ do not depend on $\lvert\psi\rangle_a$. Therefore,
\begin{equation}
\label{eq:a7}
q_k^{(\phi)}=q_k^{(\psi)},\qquad\lvert R_k^{(\phi)}\rangle_b=\lvert R_k^{(\psi)}\rangle_b,
\end{equation}
for all $\lvert\psi\rangle_a,\lvert\phi\rangle_a\in\mathcal{H}_a$ and $k\in\lbrace1,2,\dotsc,N\rbrace$.

Moreover, applying Eq.~(\ref{eq:a1}) to $\lvert\phi\rangle_a\in\mathcal{H}_a$, taking the inner product of $U\lvert\psi\rangle_a\lvert\xi\rangle_b\lvert\chi\rangle_{\pi}$ and $U\lvert\phi\rangle_a\lvert\xi\rangle_b\lvert\chi\rangle_{\pi}$, and using Eq.~(\ref{eq:a7}) we obtain
\begin{equation}
\label{eq:a8}
\langle F^{(\phi)}\vert F^{(\psi)}\rangle=\langle\phi\vert\psi\rangle.
\end{equation}
\end{proof}

\subsection{Details about the primed PBT protocol}

In this section we provide specific details about the \emph{primed} PBT protocol described in the main text, which achieves success probability $q_j$ and satisfies Eq.~(\ref{eq:5}) of the main text for the corresponding primed states,
\begin{equation}
\label{eq:b1}
\hat{\eta}_j'=\hat{\gamma}_{j,i}'=\frac{I}{2^n}.
\end{equation}

We will use the identity given in the main text, which is satisfied for any quantum state $\hat{\rho}$ of dimension $2^n$:
\begin{equation}
\label{eq:b2}
\frac{I}{2^n}\equiv \frac{1}{4^n}\sum_{l=1}^{4^n} V_l\hat{\rho} V_l^\dagger;
\end{equation}
where we define the set of unitary operations $\lbrace V_l\rbrace_{l=1}^{4^n}\equiv\lbrace\sigma_0,\sigma_1,\sigma_2,\sigma_3\rbrace^{\otimes n}$, $\sigma_0$ is the identity acting on $\mathbb{C}^2$ and $\lbrace\sigma_i\rbrace_{i=1}^3$ are the Pauli matrices.

The following, \emph{primed}, PBT protocol achieves success probability $q_j$ and satisfies Eq.~(\ref{eq:b1}). 

Firstly, consider the stage previous to the implementation of PBT in which the resource state is prepared and distributed to Alice and Bob. An ancilla $a'$ with Hilbert space $\mathcal{H}_{a'}$ of dimension $4^n$ is prepared in the state $\lvert\phi\rangle\equiv\frac{1}{2^n}\sum_{l=1}^{4^n}\lvert\mu_l\rangle$, where $\lbrace\lvert\mu_l\rangle\rbrace_{l=1}^{4^{n}}$ is an orthonormal basis of $\mathcal{H}_{a'}$. Consider the resource state $\lvert\xi\rangle_{AB}$ for the PBT protocol defined by Eqs. (\ref{eq:7})--(\ref{eq:8.5}) of the main text. In the primed protocol, the controlled unitary $\sum_{l=1}^{4^n} (\lvert\mu_l\rangle\langle\mu_l\rvert)_{a'}\bigotimes\limits_{i=1}^N (V_l)_{B_i}$ is applied on $\lvert\phi\rangle_{a'}\lvert\xi\rangle_{AB}$ in order to prepare the new resource state
\begin{equation}
\label{eq:b3}
\lvert\xi'\rangle_{a'AB}\equiv\frac{1}{2^n}\sum_{l=1}^{4^{n}}\bigotimes\limits_{i=1}^N (V_l)_{B_i}\lvert\mu_l\rangle_{a'}\lvert\xi\rangle_{AB},
\end{equation}
 where $(V_l)_{B_i}$ acts on $\mathcal{H}_{B_i}$ only. The system $a'A$ is sent to Alice and the system $B$ is sent to Bob. The initial state of system $B_j$ in this protocol is $\hat{\eta}'_j\equiv \text{Tr}_{a'A\tilde{B}_j}(\lvert\xi'\rangle\langle\xi'\rvert)_{a'AB}$. From the definitions of $\hat{\eta}_j'$, $\lvert\xi'\rangle$ and $\hat{\eta}_j$ (Eq.~(\ref{eq:8.6}) of the main text), Eq.~(\ref{eq:b2}), and the fact that $\hat{\eta}_j$ is independent of $\lvert\psi\rangle$, it is straightforward to obtain that 
\begin{equation}
\label{eq:b4}
\hat{\eta}'_j=\frac{1}{4^n}\sum_{l=1}^{4^n}V_l\hat{\eta}_j V_l^\dagger=\frac{I}{2^n},
\end{equation}
as claimed.

Now consider the implementation of the primed PBT protocol. Alice applies the unitary operation $W_{aa'}\equiv \sum_{l=1}^{4^n} (V_l^\dagger)_a\otimes(\lvert\mu_l\rangle\langle\mu_l\rvert)_{a'}$ on the system $aa'$, as the notation suggests. The global state transforms into
\begin{equation}
\label{eq:b5}
W_{aa'}\lvert\psi\rangle_a\lvert\xi\rangle_{a'AB}=\frac{1}{2^n}\sum_{l=1}^{4^{n}}(V_l^\dagger)_a\bigotimes\limits_{i=1}^N (V_l)_{B_i}\lvert\psi\rangle_a\lvert\mu_l\rangle_{a'}\lvert\xi\rangle_{AB}.
\end{equation}
Then, Alice applies her operations corresponding to the PBT protocol defined by Eqs.~(\ref{eq:7})--(\ref{eq:8.5}) of the main text on the system $aA$ only. With probability $q_j$, Alice obtains outcome $k=j\neq 0$. Due to the linearity of unitary evolution, it is not difficult to obtain that, in this case, the global state transforms into
\begin{equation}
\label{eq:b6}
\lvert G'_j\rangle_{a'aAB}=\frac{1}{2^n}\sum_{l=1}^{4^n}\bigotimes\limits_{\substack{i=1\\i\neq j}}^N (V_l)_{B_i}\lvert\mu_l\rangle_{a'}\lvert\psi\rangle_{B_j}\lvert R_j\rangle_{aA\tilde{B}_j}.
\end{equation}
Thus, we see that the state $\lvert\psi\rangle$ is teleported to port $B_j$, as required. This protocol works because, as we can see from Eq.~(\ref{eq:b5}), the operations on system $B$ commute with the operations performed by Alice on $aA$, which is necessary for satisfaction of the no-signaling principle. Therefore, this protocol is equivalent to the following: conditioned on $a'$ being in the state $\lvert\mu_l\rangle_{a'}$, if an outcome $k=j\neq 0$ is obtained, the state $V_l^{\dagger}\lvert\psi\rangle$ is teleported to port $B_j$; then, Bob applies  $\bigotimes\limits_{i=1}^N (V_l)_{B_i}$, after which, the state of system $B_j$ transforms into $\lvert\psi\rangle$, as desired.

The state of system $B_j$, after an outcome $k=i\notin\lbrace 0,j\rbrace$ is obtained, is $\hat{\gamma}'_{j,i}\equiv \text{Tr}_{a'aA\tilde{B}_j}(\lvert G'_i\rangle\langle G'_i\rvert)_{a'aAB}$.  From the definitions of $\hat{\gamma}_{j,i}'$ and $\hat{\gamma}_{j,i}$ (Eq.~(\ref{eq:8.8}) of the main text), Eq.~(\ref{eq:b2}), and the fact that $\hat{\gamma}_{j,i}$ is independent of $\lvert\psi\rangle$, it can easily be obtained that 
\begin{equation}
\label{eq:b7}
\hat{\gamma}'_{j,i}=\frac{1}{4^n}\sum_{l=1}^{4^n}V_l\hat{\gamma}_{j,i} V_l^\dagger=\frac{I}{2^n},
\end{equation}
as claimed.

If Alice obtains the outcome $k=0$, the final global state is
\begin{equation}
\label{eq:b8}
\lvert G'_0\rangle_{a'aAB}=\frac{1}{2^n}\sum_{l=1}^{4^n}\bigotimes\limits_{i=1}^N (V_l)_{B_i}\lvert\mu_l\rangle_{a'}\lvert F^{(\psi_l)}\rangle_{aAB},
\end{equation}
where $\psi_l$ refers to dependence on the state $V_l^{\dagger}\lvert\psi\rangle$.
The final state of system $B_j$ in this case is $\hat{\omega}_j'^{(\psi)}\equiv \text{Tr}_{a'aA\tilde{B}_j}\left(\lvert G_0'\rangle\langle G_0'\lvert\right)_{a'aAB}$. From the previous definition of $\hat{\omega}_j'^{(\psi)}$ and that one of $\hat{\omega}_j^{(\psi)}$, given by Eq.~(\ref{eq:8.9}) of the main text, it is straightforward to obtain that
\begin{equation}
\label{eq:b9}
\hat{\omega}_j'^{(\psi)}=\frac{1}{4^n}\sum_{l=1}^{4^n}V_l\hat{\omega}_j^{(\psi_l)}V_l^{\dagger},
\end{equation}
as given in the main text.

\begin{thebibliography}{27}%
\makeatletter
\providecommand \@ifxundefined [1]{%
 \@ifx{#1\undefined}
}%
\providecommand \@ifnum [1]{%
 \ifnum #1\expandafter \@firstoftwo
 \else \expandafter \@secondoftwo
 \fi
}%
\providecommand \@ifx [1]{%
 \ifx #1\expandafter \@firstoftwo
 \else \expandafter \@secondoftwo
 \fi
}%
\providecommand \natexlab [1]{#1}%
\providecommand \enquote  [1]{``#1''}%
\providecommand \bibnamefont  [1]{#1}%
\providecommand \bibfnamefont [1]{#1}%
\providecommand \citenamefont [1]{#1}%
\providecommand \href@noop [0]{\@secondoftwo}%
\providecommand \href [0]{\begingroup \@sanitize@url \@href}%
\providecommand \@href[1]{\@@startlink{#1}\@@href}%
\providecommand \@@href[1]{\endgroup#1\@@endlink}%
\providecommand \@sanitize@url [0]{\catcode `\\12\catcode `\$12\catcode
  `\&12\catcode `\#12\catcode `\^12\catcode `\_12\catcode `\%12\relax}%
\providecommand \@@startlink[1]{}%
\providecommand \@@endlink[0]{}%
\providecommand \url  [0]{\begingroup\@sanitize@url \@url }%
\providecommand \@url [1]{\endgroup\@href {#1}{\urlprefix }}%
\providecommand \urlprefix  [0]{URL }%
\providecommand \Eprint [0]{\href }%
\providecommand \doibase [0]{http://dx.doi.org/}%
\providecommand \selectlanguage [0]{\@gobble}%
\providecommand \bibinfo  [0]{\@secondoftwo}%
\providecommand \bibfield  [0]{\@secondoftwo}%
\providecommand \translation [1]{[#1]}%
\providecommand \BibitemOpen [0]{}%
\providecommand \bibitemStop [0]{}%
\providecommand \bibitemNoStop [0]{.\EOS\space}%
\providecommand \EOS [0]{\spacefactor3000\relax}%
\providecommand \BibitemShut  [1]{\csname bibitem#1\endcsname}%
\let\auto@bib@innerbib\@empty
\bibitem [{\citenamefont {Bennett}\ \emph {et~al.}(1993)\citenamefont
  {Bennett}, \citenamefont {Brassard}, \citenamefont {Cr\'epeau}, \citenamefont
  {Jozsa}, \citenamefont {Peres},\ and\ \citenamefont
  {Wootters}}]{teleportation}%
  \BibitemOpen
  \bibfield  {author} {\bibinfo {author} {\bibfnamefont {C.~H.}\ \bibnamefont
  {Bennett}}, \bibinfo {author} {\bibfnamefont {G.}~\bibnamefont {Brassard}},
  \bibinfo {author} {\bibfnamefont {C.}~\bibnamefont {Cr\'epeau}}, \bibinfo
  {author} {\bibfnamefont {R.}~\bibnamefont {Jozsa}}, \bibinfo {author}
  {\bibfnamefont {A.}~\bibnamefont {Peres}}, \ and\ \bibinfo {author}
  {\bibfnamefont {W.~K.}\ \bibnamefont {Wootters}},\ }\href
  {http://link.aps.org/doi/10.1103/PhysRevLett.70.1895} {\bibfield  {journal}
  {\bibinfo  {journal} {Phys. Rev. Lett.}\ }\textbf {\bibinfo {volume} {70}},\
  \bibinfo {pages} {1895} (\bibinfo {year} {1993})}\BibitemShut {NoStop}%
\bibitem [{\citenamefont {Ishizaka}\ and\ \citenamefont
  {Hiroshima}(2008)}]{IH08}%
  \BibitemOpen
  \bibfield  {author} {\bibinfo {author} {\bibfnamefont {S.}~\bibnamefont
  {Ishizaka}}\ and\ \bibinfo {author} {\bibfnamefont {T.}~\bibnamefont
  {Hiroshima}},\ }\href
  {http://link.aps.org/doi/10.1103/PhysRevLett.101.240501} {\bibfield
  {journal} {\bibinfo  {journal} {Phys. Rev. Lett.}\ }\textbf {\bibinfo
  {volume} {101}},\ \bibinfo {pages} {240501} (\bibinfo {year}
  {2008})}\BibitemShut {NoStop}%
\bibitem [{\citenamefont {Ishizaka}\ and\ \citenamefont
  {Hiroshima}(2009)}]{IH09}%
  \BibitemOpen
  \bibfield  {author} {\bibinfo {author} {\bibfnamefont {S.}~\bibnamefont
  {Ishizaka}}\ and\ \bibinfo {author} {\bibfnamefont {T.}~\bibnamefont
  {Hiroshima}},\ }\href {http://link.aps.org/doi/10.1103/PhysRevA.79.042306}
  {\bibfield  {journal} {\bibinfo  {journal} {Phys. Rev. A}\ }\textbf {\bibinfo
  {volume} {79}},\ \bibinfo {pages} {042306} (\bibinfo {year}
  {2009})}\BibitemShut {NoStop}%
\bibitem [{\citenamefont {Nielsen}\ and\ \citenamefont {Chuang}(1997)}]{NC97}%
  \BibitemOpen
  \bibfield  {author} {\bibinfo {author} {\bibfnamefont {M.~A.}\ \bibnamefont
  {Nielsen}}\ and\ \bibinfo {author} {\bibfnamefont {I.~L.}\ \bibnamefont
  {Chuang}},\ }\href {http://link.aps.org/doi/10.1103/PhysRevLett.79.321}
  {\bibfield  {journal} {\bibinfo  {journal} {Phys. Rev. Lett.}\ }\textbf
  {\bibinfo {volume} {79}},\ \bibinfo {pages} {321} (\bibinfo {year}
  {1997})}\BibitemShut {NoStop}%
\bibitem [{\citenamefont {Beigi}\ and\ \citenamefont
  {K{\"{o}}nig}(2011)}]{BK11}%
  \BibitemOpen
  \bibfield  {author} {\bibinfo {author} {\bibfnamefont {S.}~\bibnamefont
  {Beigi}}\ and\ \bibinfo {author} {\bibfnamefont {R.}~\bibnamefont
  {K{\"{o}}nig}},\ }\href {http://iopscience.iop.org/1367-2630/13/9/093036/}
  {\bibfield  {journal} {\bibinfo  {journal} {New J. Phys.}\ }\textbf {\bibinfo
  {volume} {13}},\ \bibinfo {pages} {093036} (\bibinfo {year}
  {2011})}\BibitemShut {NoStop}%
\bibitem [{\citenamefont {Vaidman}(2003)}]{V03}%
  \BibitemOpen
  \bibfield  {author} {\bibinfo {author} {\bibfnamefont {L.}~\bibnamefont
  {Vaidman}},\ }\href {http://link.aps.org/doi/10.1103/PhysRevLett.90.010402}
  {\bibfield  {journal} {\bibinfo  {journal} {Phys. Rev. Lett.}\ }\textbf
  {\bibinfo {volume} {90}},\ \bibinfo {pages} {010402} (\bibinfo {year}
  {2003})}\BibitemShut {NoStop}%
\bibitem [{\citenamefont {Kent}\ \emph {et~al.}(2006)\citenamefont {Kent},
  \citenamefont {Beausoleil}, \citenamefont {Munro},\ and\ \citenamefont
  {Spiller}}]{KBMS06}%
  \BibitemOpen
  \bibfield  {author} {\bibinfo {author} {\bibfnamefont {A.}~\bibnamefont
  {Kent}}, \bibinfo {author} {\bibfnamefont {R.}~\bibnamefont {Beausoleil}},
  \bibinfo {author} {\bibfnamefont {W.}~\bibnamefont {Munro}}, \ and\ \bibinfo
  {author} {\bibfnamefont {T.}~\bibnamefont {Spiller}},\ }\href@noop {}
  {}\bibinfo {howpublished} {{US Patent No.} 7,075,438} (\bibinfo {year}
  {2006})\BibitemShut {NoStop}%
\bibitem [{\citenamefont {Malaney}(2010{\natexlab{a}})}]{M10.1}%
  \BibitemOpen
  \bibfield  {author} {\bibinfo {author} {\bibfnamefont {R.~A.}\ \bibnamefont
  {Malaney}},\ }\href {http://link.aps.org/doi/10.1103/PhysRevA.81.042319}
  {\bibfield  {journal} {\bibinfo  {journal} {Phys. Rev. A}\ }\textbf {\bibinfo
  {volume} {81}},\ \bibinfo {pages} {042319} (\bibinfo {year}
  {2010}{\natexlab{a}})}\BibitemShut {NoStop}%
\bibitem [{\citenamefont {Malaney}(2010{\natexlab{b}})}]{M10.2}%
  \BibitemOpen
  \bibfield  {author} {\bibinfo {author} {\bibfnamefont {R.~A.}\ \bibnamefont
  {Malaney}},\ }\href@noop {} {} (\bibinfo {year} {2010}{\natexlab{b}}),\
  \Eprint {http://arxiv.org/abs/1004.4689} {arXiv:1004.4689} \BibitemShut
  {NoStop}%
\bibitem [{\citenamefont {Kent}\ \emph {et~al.}(2011)\citenamefont {Kent},
  \citenamefont {Munro},\ and\ \citenamefont {Spiller}}]{KMS11}%
  \BibitemOpen
  \bibfield  {author} {\bibinfo {author} {\bibfnamefont {A.}~\bibnamefont
  {Kent}}, \bibinfo {author} {\bibfnamefont {W.~J.}\ \bibnamefont {Munro}}, \
  and\ \bibinfo {author} {\bibfnamefont {T.~P.}\ \bibnamefont {Spiller}},\
  }\href {http://link.aps.org/doi/10.1103/PhysRevA.84.012326} {\bibfield
  {journal} {\bibinfo  {journal} {Phys. Rev. A}\ }\textbf {\bibinfo {volume}
  {84}},\ \bibinfo {pages} {012326} (\bibinfo {year} {2011})}\BibitemShut
  {NoStop}%
\bibitem [{\citenamefont {Lau}\ and\ \citenamefont {Lo}(2011)}]{LL11}%
  \BibitemOpen
  \bibfield  {author} {\bibinfo {author} {\bibfnamefont {H.~K.}\ \bibnamefont
  {Lau}}\ and\ \bibinfo {author} {\bibfnamefont {H.~K.}\ \bibnamefont {Lo}},\
  }\href {http://link.aps.org/doi/10.1103/PhysRevA.83.012322} {\bibfield
  {journal} {\bibinfo  {journal} {Phys. Rev. A}\ }\textbf {\bibinfo {volume}
  {83}},\ \bibinfo {pages} {012322} (\bibinfo {year} {2011})}\BibitemShut
  {NoStop}%
\bibitem [{\citenamefont {Buhrman}\ \emph {et~al.}(2011)\citenamefont
  {Buhrman}, \citenamefont {Chandran}, \citenamefont {Fehr}, \citenamefont
  {Gelles}, \citenamefont {Goyal}, \citenamefont {Ostrovsky},\ and\
  \citenamefont {Schaffner}}]{pbqc}%
  \BibitemOpen
  \bibfield  {author} {\bibinfo {author} {\bibfnamefont {H.}~\bibnamefont
  {Buhrman}}, \bibinfo {author} {\bibfnamefont {N.}~\bibnamefont {Chandran}},
  \bibinfo {author} {\bibfnamefont {S.}~\bibnamefont {Fehr}}, \bibinfo {author}
  {\bibfnamefont {R.}~\bibnamefont {Gelles}}, \bibinfo {author} {\bibfnamefont
  {V.}~\bibnamefont {Goyal}}, \bibinfo {author} {\bibfnamefont
  {R.}~\bibnamefont {Ostrovsky}}, \ and\ \bibinfo {author} {\bibfnamefont
  {C.}~\bibnamefont {Schaffner}},\ }\href@noop {} {} (\bibinfo {year} {2011}),\
  \Eprint {http://arxiv.org/abs/1009.2490} {arXiv:1009.2490} \BibitemShut
  {NoStop}%
\bibitem [{\citenamefont {Buhrman}\ \emph {et~al.}()\citenamefont {Buhrman},
  \citenamefont {Fehr}, \citenamefont {Schaffner},\ and\ \citenamefont
  {Speelman}}]{BFSS11}%
  \BibitemOpen
  \bibfield  {author} {\bibinfo {author} {\bibfnamefont {H.}~\bibnamefont
  {Buhrman}}, \bibinfo {author} {\bibfnamefont {S.}~\bibnamefont {Fehr}},
  \bibinfo {author} {\bibfnamefont {C.}~\bibnamefont {Schaffner}}, \ and\
  \bibinfo {author} {\bibfnamefont {F.}~\bibnamefont {Speelman}},\ }\href@noop
  {} {}\bibinfo {howpublished} {in \emph{Proceedings of the 4th Conference on
  Innovations in Theoretical Computer Science} (ACM, New York, 2013), pp.
  145--158}\BibitemShut {NoStop}%
\bibitem [{\citenamefont {Groisman}\ and\ \citenamefont {Reznik}(2002)}]{GR02}%
  \BibitemOpen
  \bibfield  {author} {\bibinfo {author} {\bibfnamefont {B.}~\bibnamefont
  {Groisman}}\ and\ \bibinfo {author} {\bibfnamefont {B.}~\bibnamefont
  {Reznik}},\ }\href {http://link.aps.org/doi/10.1103/PhysRevA.66.022110}
  {\bibfield  {journal} {\bibinfo  {journal} {Phys. Rev. A}\ }\textbf {\bibinfo
  {volume} {66}},\ \bibinfo {pages} {022110} (\bibinfo {year}
  {2002})}\BibitemShut {NoStop}%
\bibitem [{\citenamefont {Groisman}\ \emph {et~al.}(2003)\citenamefont
  {Groisman}, \citenamefont {Reznik},\ and\ \citenamefont {Vaidman}}]{GRV03}%
  \BibitemOpen
  \bibfield  {author} {\bibinfo {author} {\bibfnamefont {B.}~\bibnamefont
  {Groisman}}, \bibinfo {author} {\bibfnamefont {B.}~\bibnamefont {Reznik}}, \
  and\ \bibinfo {author} {\bibfnamefont {L.}~\bibnamefont {Vaidman}},\ }\href
  {http://www.tandfonline.com/doi/abs/10.1080/09500340308234543} {\bibfield
  {journal} {\bibinfo  {journal} {J. Mod. Opt.}\ }\textbf {\bibinfo {volume}
  {50}},\ \bibinfo {pages} {943} (\bibinfo {year} {2003})}\BibitemShut
  {NoStop}%
\bibitem [{\citenamefont {Clark}\ \emph {et~al.}(2010)\citenamefont {Clark},
  \citenamefont {Connor}, \citenamefont {Jaksch},\ and\ \citenamefont
  {Popescu}}]{CCJP10}%
  \BibitemOpen
  \bibfield  {author} {\bibinfo {author} {\bibfnamefont {S.~R.}\ \bibnamefont
  {Clark}}, \bibinfo {author} {\bibfnamefont {A.~J.}\ \bibnamefont {Connor}},
  \bibinfo {author} {\bibfnamefont {D.}~\bibnamefont {Jaksch}}, \ and\ \bibinfo
  {author} {\bibfnamefont {S.}~\bibnamefont {Popescu}},\ }\href
  {http://iopscience.iop.org/1367-2630/12/8/083034} {\bibfield  {journal}
  {\bibinfo  {journal} {New J. Phys.}\ }\textbf {\bibinfo {volume} {12}},\
  \bibinfo {pages} {083034} (\bibinfo {year} {2010})}\BibitemShut {NoStop}%
\bibitem [{\citenamefont {Kent}(2011)}]{K11.1}%
  \BibitemOpen
  \bibfield  {author} {\bibinfo {author} {\bibfnamefont {A.}~\bibnamefont
  {Kent}},\ }\href {http://link.aps.org/doi/10.1103/PhysRevA.84.022335}
  {\bibfield  {journal} {\bibinfo  {journal} {Phys. Rev. A}\ }\textbf {\bibinfo
  {volume} {84}},\ \bibinfo {pages} {022335} (\bibinfo {year}
  {2011})}\BibitemShut {NoStop}%
\bibitem [{\citenamefont {Gisin}(1998)}]{G98}%
  \BibitemOpen
  \bibfield  {author} {\bibinfo {author} {\bibfnamefont {N.}~\bibnamefont
  {Gisin}},\ }\href {http://dx.doi.org/10.1016/S0375-9601(98)00170-4}
  {\bibfield  {journal} {\bibinfo  {journal} {Phys. Lett. A}\ }\textbf
  {\bibinfo {volume} {242}},\ \bibinfo {pages} {1} (\bibinfo {year}
  {1998})}\BibitemShut {NoStop}%
\bibitem [{\citenamefont {Barrett}\ \emph {et~al.}(2005)\citenamefont
  {Barrett}, \citenamefont {Hardy},\ and\ \citenamefont {Kent}}]{BHK05}%
  \BibitemOpen
  \bibfield  {author} {\bibinfo {author} {\bibfnamefont {J.}~\bibnamefont
  {Barrett}}, \bibinfo {author} {\bibfnamefont {L.}~\bibnamefont {Hardy}}, \
  and\ \bibinfo {author} {\bibfnamefont {A.}~\bibnamefont {Kent}},\ }\href
  {http://link.aps.org/doi/10.1103/PhysRevLett.95.010503} {\bibfield  {journal}
  {\bibinfo  {journal} {Phys. Rev. Lett.}\ }\textbf {\bibinfo {volume} {95}},\
  \bibinfo {pages} {010503} (\bibinfo {year} {2005})}\BibitemShut {NoStop}%
\bibitem [{\citenamefont {Paw{\l}owski}\ \emph {et~al.}(2009)\citenamefont
  {Paw{\l}owski}, \citenamefont {Paterek}, \citenamefont {Kaszlikowski},
  \citenamefont {Scarani}, \citenamefont {Winter},\ and\ \citenamefont
  {{\.{Z}}ukowski}}]{ic}%
  \BibitemOpen
  \bibfield  {author} {\bibinfo {author} {\bibfnamefont {M.}~\bibnamefont
  {Paw{\l}owski}}, \bibinfo {author} {\bibfnamefont {T.}~\bibnamefont
  {Paterek}}, \bibinfo {author} {\bibfnamefont {D.}~\bibnamefont
  {Kaszlikowski}}, \bibinfo {author} {\bibfnamefont {V.}~\bibnamefont
  {Scarani}}, \bibinfo {author} {\bibfnamefont {A.}~\bibnamefont {Winter}}, \
  and\ \bibinfo {author} {\bibfnamefont {M.}~\bibnamefont {{\.{Z}}ukowski}},\
  }\href
  {http://www.nature.com/nature/journal/v461/n7267/full/nature08400.html}
  {\bibfield  {journal} {\bibinfo  {journal} {Nature (London)}\ }\textbf
  {\bibinfo {volume} {461}},\ \bibinfo {pages} {1101} (\bibinfo {year}
  {2009})}\BibitemShut {NoStop}%
\bibitem [{\citenamefont {Bennett}\ and\ \citenamefont {Wiesner}(1992)}]{sdc}%
  \BibitemOpen
  \bibfield  {author} {\bibinfo {author} {\bibfnamefont {C.~H.}\ \bibnamefont
  {Bennett}}\ and\ \bibinfo {author} {\bibfnamefont {S.~J.}\ \bibnamefont
  {Wiesner}},\ }\href {http://link.aps.org/doi/10.1103/PhysRevLett.69.2881}
  {\bibfield  {journal} {\bibinfo  {journal} {Phys. Rev. Lett.}\ }\textbf
  {\bibinfo {volume} {69}},\ \bibinfo {pages} {2881} (\bibinfo {year}
  {1992})}\BibitemShut {NoStop}%
\bibitem [{Note1()}]{Note1}%
  \BibitemOpen
  \bibinfo {note} {It is not possible that $p_j'<4^{-n}$, because this would
  imply that a modified protocol in which Bob applies a permutation to the
  obtained message succeeds with a probability higher than $4^{-n}$, violating
  the no-signaling principle.}\BibitemShut {Stop}%
\bibitem [{\citenamefont {Wootters}\ and\ \citenamefont {Zurek}(1982)}]{WZ82}%
  \BibitemOpen
  \bibfield  {author} {\bibinfo {author} {\bibfnamefont {W.~K.}\ \bibnamefont
  {Wootters}}\ and\ \bibinfo {author} {\bibfnamefont {W.~H.}\ \bibnamefont
  {Zurek}},\ }\href
  {http://www.nature.com/nature/journal/v299/n5886/abs/299802a0.html}
  {\bibfield  {journal} {\bibinfo  {journal} {Nature (London)}\ }\textbf
  {\bibinfo {volume} {299}},\ \bibinfo {pages} {802} (\bibinfo {year}
  {1982})}\BibitemShut {NoStop}%
\bibitem [{\citenamefont {Dieks}(1982)}]{D82}%
  \BibitemOpen
  \bibfield  {author} {\bibinfo {author} {\bibfnamefont {D.}~\bibnamefont
  {Dieks}},\ }\href {http://dx.doi.org/10.1016/0375-9601(82)90084-6} {\bibfield
   {journal} {\bibinfo  {journal} {Phys. Lett. A}\ }\textbf {\bibinfo {volume}
  {92}},\ \bibinfo {pages} {271} (\bibinfo {year} {1982})}\BibitemShut
  {NoStop}%
\bibitem [{\citenamefont {Barrett}(2007)}]{B07}%
  \BibitemOpen
  \bibfield  {author} {\bibinfo {author} {\bibfnamefont {J.}~\bibnamefont
  {Barrett}},\ }\href {http://link.aps.org/doi/10.1103/PhysRevA.75.032304}
  {\bibfield  {journal} {\bibinfo  {journal} {Phys. Rev. A}\ }\textbf {\bibinfo
  {volume} {75}},\ \bibinfo {pages} {032304} (\bibinfo {year}
  {2007})}\BibitemShut {NoStop}%
\bibitem [{\citenamefont {Jozsa}(2002)}]{J02}%
  \BibitemOpen
  \bibfield  {author} {\bibinfo {author} {\bibfnamefont {R.}~\bibnamefont
  {Jozsa}},\ }\href@noop {} {} (\bibinfo {year} {2002}),\ \Eprint
  {http://arxiv.org/abs/quant-ph/0204153} {arXiv:quant-ph/0204153} \BibitemShut
  {NoStop}%
\bibitem [{\citenamefont {Jozsa}(2003)}]{J03}%
  \BibitemOpen
  \bibfield  {author} {\bibinfo {author} {\bibfnamefont {R.}~\bibnamefont
  {Jozsa}},\ }\href@noop {} {} (\bibinfo {year} {2003}),\ \Eprint
  {http://arxiv.org/abs/quant-ph/0305114} {arXiv:quant-ph/0305114} \BibitemShut
  {NoStop}%
\end{thebibliography}
\end{document}